\documentclass[11pt,letterpaper]{article}
\usepackage[top=80pt,bottom=80pt,left=88pt,right=86pt]{geometry}
\usepackage{amssymb}
\usepackage{graphicx}
\usepackage{algorithm}
\usepackage[noend]{algorithmic}
\usepackage{microtype}
\usepackage{subfig}
\usepackage{amsmath}
\usepackage{amsthm}
\usepackage{enumerate}
\usepackage{wrapfig}
\usepackage{multirow}
\usepackage{nicefrac}
\usepackage[colorinlistoftodos,prependcaption,textsize=tiny]{todonotes}
\usepackage[colorlinks,linkcolor=blue,filecolor=blue,citecolor=blue,urlcolor=blue,pdfstartview=FitH,pagebackref]{hyperref}
\usepackage[nameinlink]{cleveref}
\usepackage{authblk}

\usepackage{url}

\definecolor{forestgreen}{rgb}{0.13, 0.55, 0.13}

\newtheorem{theorem}{Theorem}
\newtheorem{lemma}[theorem]{Lemma}
\newtheorem{claim}[theorem]{Claim}

\newtheorem{observation}[theorem]{Observation}

\newtheorem{corollary}[theorem]{Corollary}

\graphicspath{{./figures/}}

\def\reals{{\mathbb R}}

\def\eps{{\varepsilon}}

\newcommand{\opt}{\text{OPT}}
\newcommand{\vis}{\textit{Vis}}
\newcommand{\seg}[1]{\overline{#1}}
\newcommand{\uv}{\seg{uv}}
\newcommand{\xy}{\seg{xy}}
\newcommand{\xyp}{\seg{x'y'}}

\def\dfn#1{\emph{\textcolor{forestgreen}{\textbf{#1}}}}

\date{}

\begin{document}

\title{A Constant-Factor Approximation Algorithm for\\ Vertex Guarding a WV-Polygon}

\author{Stav Ashur}
\author{Omrit Filtser}
\author{Matthew J. Katz}
\affil{Ben-Gurion University of the Negev\\ Email: \texttt{\{stavshe,omritna,matya\}@cs.bgu.ac.il}}

\maketitle

\begin{abstract}
The problem of vertex guarding a simple polygon was first studied by Subir K. Ghosh (1987), who presented a polynomial-time $O(\log n)$-approximation algorithm for placing as few guards as possible at vertices of a simple $n$-gon $P$, such that every point in $P$ is visible to at least one of the guards. 
Ghosh also conjectured that this problem admits a polynomial-time algorithm with constant approximation ratio. Due to the centrality of guarding problems in the field of computational geometry, much effort has been invested throughout the years in trying to resolve this conjecture. Despite some progress (surveyed below), the conjecture remains unresolved to date.
In this paper, we confirm the conjecture for the important case of weakly visible polygons, by presenting a $(2+\eps)$-approximation algorithm for guarding such a polygon using vertex guards. A simple polygon $P$ is \emph{weakly visible} if it has an edge $e$, such that every point in $P$ is visible from some point on $e$. We also present a $(2+\eps)$-approximation algorithm for guarding a weakly visible polygon $P$, where guards may be placed anywhere on $P$'s boundary (except in the interior of the edge $e$). Finally, we present a $3c$-approximation algorithm for vertex guarding a polygon $P$ that is weakly visible from a chord, given a subset $G$ of $P$'s vertices that guards $P$'s boundary whose size is bounded by $c$ times the size of a minimum such subset.

Our algorithms are based on an in-depth analysis of the geometric properties of the regions that remain unguarded after placing guards at the vertices to guard the polygon's boundary. It is plausible that our results will enable Bhattacharya et al. to complete their grand attempt to prove the original conjecture, as their approach is based on partitioning the underlying simple polygon into a hierarchy of weakly visible polygons.
\end{abstract}

\section{Introduction}
The Art Gallery Problem is a classical problem in computational geometry, posed by Victor Klee in 1973: Place a minimum number of points (representing guards) in a given simple polygon $P$ (representing an art gallery), so that every point in $P$ is seen by at least one of the points. We say that a point $p \in P$ \emph{sees} (or \emph{guards}) a point $q \in P$ if the line segment $\seg{pq}$ is contained in $P$. We say that a subset $G \subseteq P$ \emph{guards} $P$, if every point in $P$ is seen by at least one point in $G$.

There are numerous variants of the art gallery problem, which are also referred to as the art gallery problem. These variants differ from one another in (i) the underlying domain, e.g., simple polygon, polygon with holes, orthogonal polygon, or terrain, (ii) which parts of the domain must be guarded, e.g., only its vertices, only its boundary, or the entire domain, (iii) the type of guards, e.g., static, mobile, or with various restrictions on their coverage area such as limited range, (iv) the restrictions on the location of the guards, e.g., only at vertices (vertex-guards), only on the boundary (boundary-guards), or anywhere (point-guards), and (v) the underlying notion of visibility, e.g., line of sight, rectangle visibility, or staircase visibility. It is impossible to survey here the vast literature on the art gallery problem, ranging from combinatorial and optimization results to hardness of computation results, so we only mention the book by O'rourke~\cite{O'rourke87} and a small sample of recent papers~\cite{AAM18,BM17}. 

In this paper, we deal with the version of the art gallery problem, where the guards are confined to the boundary of the underlying polygon, and in particular to its vertices. Such guards are referred to as \emph{boundary} guards or \emph{vertex} guards, respectively. The first to present results for this version was Ghosh~\cite{Ghosh87,Ghosh10}, who gave a polynomial-time $O(\log n)$-approximation algorithm for guarding either a simple polygon or a polygon with holes using vertex (or edge) guards. In the related work paragraph below we survey many of the subsequent results for this version. 

In this paper, we consider a special family of simple polygons, namely, the family of \emph{weakly visible} polygons. A simple polygon $P$ is \emph{weakly visible} if it has an edge $e$, such that every point in $P$ is visible from some point on $e$, or, in other words, a guard patrolling along $e$ can see the entire polygon. We also consider polygons that are weakly visible from a chord, rather than an edge, where a \emph{chord} in a polygon $P$ is a line segment whose endpoints are on the boundary of $P$ and whose interior is contained in the interior of $P$.

The problem of guarding a weakly visible polygon (WV-polygon) $P$ by vertex guards was studied by Bhattacharya et al.~\cite{BGR17}. They first present a 4-approximation algorithm for vertex guarding \emph{only} the vertices of $P$. Next, they claim that this algorithm places the guards at vertices in such a way that each of the remaining unguarded regions of $P$ has a boundary edge which is contained in an edge of $P$. Based on this claim, they devise a 6-approximation algorithm for vertex guarding $P$'s boundary, and present it as a 6-approximation algorithm for guarding $P$ (boundary plus interior). Unfortunately, this claim is false; counterexamples were constructed and approved by the authors (who are now attempting to fix their algorithm, so as to obtain an algorithm for vertex guarding $P$ entirely)~\cite{BGPR18}. Thus, the challenge of obtaining a constant-factor approximation algorithm for guarding a WV-polygon with vertex guards or boundary guards is still on. 

The main result of this paper is such an algorithm. Specifically, 
denote by \opt\ the size of a minimum-cardinality subset of the vertices of $P$ that guards $P$. 
We present a polynomial-time algorithm that finds a subset $I$ of the vertices of $P$, such that $I$ guards $P$ and $|I| \le (2+\eps)\opt$, for any constant $\eps > 0$.

Already in 1987, Ghosh conjectured that there exists a constant-factor approximation algorithm for vertex guarding a simple polygon. Recently, Bhattacharya et al.~\cite{bhattacharya2017constant} managed to devise such an algorithm for vertex guarding the \emph{vertices} of a simple polygon $P$, by first partitioning $P$ into a hierarchy of weakly visible polygons according to the link distance from a starting vertex. They also present such an algorithm for vertex guarding $P$ (boundary plus interior), however, this algorithm is erroneous, since it relies on the false statement mentioned above. Thus, our result is a significant step towards resolving Ghosh's conjecture, and may also provide the missing ingredient for Bhattacharya et al. to fully resolve the conjecture.  

Prior to our result, the only (non-trivial) family of polygons for which a constant-factor approximation algorithm for guarding a member of the family was known, is the family of monotone algorithms. Specifically, Krohn and Nilsson~\cite{KN13} showed that vertex guarding a monotone polygon is NP-hard, and presented a constant-factor approximation algorithm for \emph{point} guarding such a polygon (as well as an $O(\opt^2)$-approximation algorithm for point guarding a rectilinear polygon). 

\paragraph{Related work.}

Several improvements to Ghosh's $O(\log n)$-approximation algorithm were obtained over the years.
In 2006, Efrat and Har-Peled~\cite{EH06} described a randomized polynomial-time $O(\log \opt)$-approximation algorithm for vertex guarding a simple polygon, where the approximation factor is correct with high probability. 
Moreover, they considered the version in which the guards are restricted to the points of an arbitrarily dense grid, and presented a randomized algorithm which returns an $O(\log \opt)$-approximation with high probability, where $\opt$ is the size of an optimal solution to the modified problem and the running time depends on the ratio between the diameter of the polygon and the grid size. In 2007, Deshpande et al.~\cite{DKDS07} presented a deterministic $O(\log \opt)$-approximation algorithm for (arbitrary) point guards, with running time polynomial in $n$ and the spread of the vertices\footnote{The spread of a set of points is the ratio between the longest and shortest pairwise distances.}. Combining ideas from the latter two algorithms, Bonnet and Miltzow~\cite{BM17} presented a randomized polynomial-time $O(\log \opt)$-approximation algorithm for point guards, assuming integer coordinates and a general position assumption on the vertices of the underlying simple polygon. 

In 2011, King and Kirkpatrick~\cite{KK11} observed that by applying methods that were developed for the Set Cover problem after the publication of Ghosh's algorithm, one can obtain an $O(\log \opt)$ approximation factor for vertex guarding a simple polygon (and an $O(\log h \log \opt)$ factor for vertex guarding a polygon with $h$ holes). Moreover, they improved the approximation factor to $O(\log \log \opt)$ for guarding a simple polygon either with vertex guards or boundary guards, where in the former case the running time is polynomial in $n$ and in the latter case it is polynomial in $n$ and the spread of the vertices.

Most of the variants of the Art Gallery Problem are NP-hard. O'Rourke and Supowit~\cite{ORourkeS83} proved this for polygons with holes and  point guards, Lee and Lin~\cite{LL86} proved this for simple polygons and vertex guards, and Aggarwal~\cite{Aggarwal84} generalized the latter proof to simple polygons and point guards. Eidenbenz et al.~\cite{EidenbenzSW01} presented a collection of hardness results. In particular, they proved that it is unlikely that a PTAS exists for vertex guarding or point guarding a simple polygon. Recently, Abrahamsen et al.~\cite{AAM18} proved that the Art Gallery Problem is $\exists\mathbb{R}$-complete, for simple polygons and point guards.

Weakly Visible polygons were defined and studied in the context of mobile guards~\cite{O'rourke87}.
An \emph{edge guard} is a guard that moves along an edge of the polygon. Thus, a simple polygon is weakly visible if and only if it can be guarded by a single edge guard, and the problem of vertex guarding a WV-polygon is equivalent to the problem of replacing the single edge guard by as few (static) vertex guards as possible. Avis and Toussaint~\cite{AT81} presented a linear-time algorithm for detecting whether a polygon is weakly visible from a given edge. Subsequently, Sack and Suri~\cite{SS88} and Das et al.~\cite{DHN94}
devised linear-time algorithms which output all the edges (if any) from which the polygon is weakly visible. Algorithms for finding either an edge or a chord from which the polygon is weakly visible were given by Ke~\cite{Ke87} and by Ghosh et al.~\cite{GMPSM93}.
Finally, Bhattacharya et al~\cite{BGR17} proved that the problem of point guarding a WV-polygon is NP-hard, and that
there does not exist a polynomial-time algorithm for vertex guarding a WV-polygon with holes with approximation factor better than $((1-\eps)/12)\ln n$, unless P $=$ NP. 

Our algorithm for vertex guarding a WV-polygon uses a solution to the problem of guarding the boundary of a WV-polygon using vertex guards. This problem admits a local-search-based PTAS (see~\cite{Katz18}), which is similar to the local-search-based PTAS of Gibson et al.~\cite{Gibson14} for vertex guarding the vertices of a 1.5D-terrain. The proof of both these PTASs is based on the proof scheme of Mustafa and Ray~\cite{MR09}.

\paragraph{Results.}
Our algorithm for vertex guarding a WV-polygon $P$ (presented in \Cref{sec:algorithm}) consists of two main parts. In the first part, it computes a subset $G$ of the vertices of $P$ that guards $P$'s boundary. This is done by applying a known algorithm for this task. In the second part, it computes a subset $G'$ of the vertices of $P$ of size at most that of $G$, such that $G \cup G'$ guards $P$. Thus, if we apply the algorithm of \cite{Katz18} for computing $G$, then the approximation ratio of our algorithm is $2+\eps$, since the algorithm of Ashur et al. guarantees that $|G| \le (1 + \eps/2)\opt_\partial$, where $\opt_\partial$ is the size of a minimum-cardinality subset of the vertices of $P$ that guards $P$'s boundary, and clearly $\opt_\partial \le \opt$.       

Let $x$ be a vertex in $G$ and let $\vis(x)$ be the visibility polygon of $x$ (i.e., $\vis(x)$ is the set of all points of $P$ that are visible from $x$), then $P \setminus Vis(x)$ is a set of connected regions, which we refer to as \emph{pockets}. Moreover, 
a connected subset $H$ of $P$ is a \emph{hole} in $P$ w.r.t. $G$ if (i) there is no point in $H$ that is visible from $G$, and (ii) $H$ is maximal in the sense that any connected subset of $P$ that strictly contains $H$ has a point that is visible from $G$. The second part of our algorithm (and its proof) are based on a deep structural analysis and characterization of the pockets and holes in $P$ (presented in \Cref{sec:structure}).

The requirement that $G$ is a subset of the vertices of $P$ is actually not necessary; the second part of our algorithm only needs a set of boundary points that guards $P$'s boundary. This observation enables us to use a smaller number of guards, assuming that boundary guards are allowed.

Finally, in \Cref{sec:chord}, we consider the more general family of polygons, those that are weakly visible from a chord. Notice that a chord $\uv$ in $P$ slices $P$ into two (sub)polygons, such that each of them is weakly visible w.r.t. the edge $\uv$. After updating two of the geometric claims presented in \Cref{sec:structure}, we show how to apply our algorithm to a polygon that is weakly visible from a chord. The approximation ratio in this case is $3|G|$ (rather than $2|G|$). However, the best known algorithm for computing $G$ in this case is that of Bachattarya et al.~\cite{bhattacharya2017constant}, which computes a $c$-approximation, for some constant $c$. Therefore, the approximation ratio of our algorithm in this case is $3c$.

\section{Structural analysis}
\label{sec:structure}

For two points $x,y\in \reals^2$, we denote by $\xy$ the line segment whose endpoints are $x$ and $y$, and by $\ell_{xy}$ the line through $x$ and $y$. We denote by $\ell_x$ the horizontal line through $x$.

Let $P$ be a polygon whose set of vertices is $V=\{u=v_1,v_2,\dots,v_n=v\}$, and which is weakly visible from its edge $e=\uv$. We denote the boundary of $P$ by $\partial P$.
The edges of $P$ are the segments $\seg{v_1v_2},\seg{v_2v_3},\dots,\seg{v_{n-1}v_n}$ and $\seg{v_1v_n}=\uv$.
We assume w.l.o.g. that $\uv$ is contained in the x-axis, and $u$ is to the left of $v$. 
Furthermore, we assume that $P$ is contained in the (closed) halfplane above the x-axis; in particular, the angles at $u$ and $v$ are convex. This assumption can be easily removed, as we show towards the end of \Cref{sec:algorithm}.

\subsection{Visibility polygons}
For a point $p\in P$, let $\vis(p)=\{q \, \mid \, \seg{pq}\subseteq P\}$ be the \dfn{visibility polygon} of $p$. In other words, $\vis(p)$ is the set of all points of $P$ that are visible from $p$. Notice that $\vis(p)$ is a star-shaped polygon, and thus clearly also a simple polygon, contained in $P$ (see \Cref{fig:visiblityPolygon}).

\begin{figure}[h]
	\centering
	\includegraphics[scale=1]{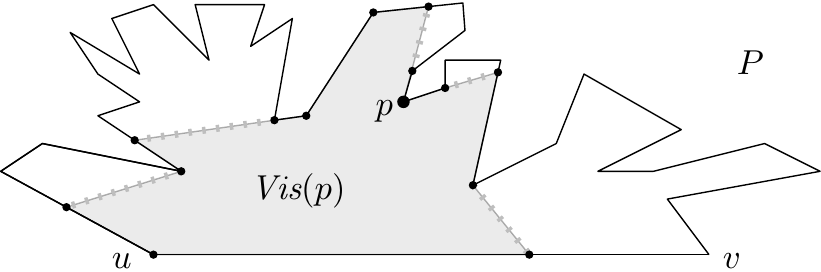}
	\caption{The visibility polygon of a point $p\in P$.}
	\label{fig:visiblityPolygon}
\end{figure}

Any vertex of $P$ that belongs to $\vis(p)$ is also considered a vertex of $\vis(p)$. Consider the set of edges of $\vis(p)$. Some of these edges are fully contained in $\partial P$. The edges that are not contained in $\partial P$ are \dfn{constructed edges}: these are edges whose endpoints are on $\partial P$ and whose interior is contained in the interior of $P$ (these are the gray dotted edges in \Cref{fig:visiblityPolygon}).

\begin{claim}\label{clm:singleSegmentOfUV}
	For any $p\in P$, there exists a single edge of $\vis(p)$ that is contained in $\uv$.
\end{claim}
\begin{proof}
	Since $p$ is visible from $\uv$, $\uv\cap\vis(p)\neq\emptyset$. Let $a$ (resp. $b$) be the leftmost (resp. rightmost) point on $\uv$ that belongs to $\vis(p)$.
	Assume by contradiction that there exists a point $c$ on $\seg{ab}$ that is not visible from $p$. Then, there exists a point $d$ of $\partial P$ in the interior of the segment $\seg{pc}$ (see \Cref{fig:singleSegmentOfuv}). The triangle $\triangle pab$ does not contain any points from $\partial P$ in its interior, because $\partial P$ cannot cross the segments $\seg{pa}$ and $\seg{bp}$. But $\seg{pc}\subseteq\triangle pab$, and thus $d$ is in the interior of the triangle --- a contradiction.
\end{proof}

\begin{figure}[h]
	\centering
	\includegraphics[scale=1]{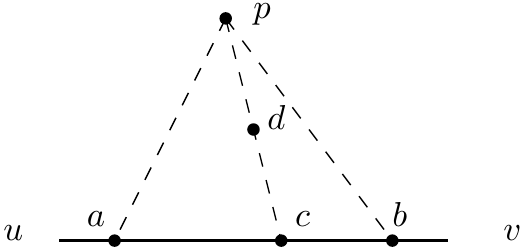}
	\caption{If $p$ sees both $a$ and $b$, then $p$ sees the entire segment $\seg{ab}$.}
	\label{fig:singleSegmentOfuv}
\end{figure}

\subsection{Pockets}
Consider $P\setminus Vis(p)$ (see \Cref{fig:visiblityPolygon}). This is a set of connected regions, which we refer to as \dfn{pockets}. Since $P$ is a simple polygon, each pocket $C$ is adjacent to a single constructed edge, and therefore the intersection of $\partial P$ and $C$ is connected. We refer to $\partial P\cap C$ as the boundary of the pocket $C$, and denote it by $\partial C$. (Thus the constructed edge itself is not part of $\partial C$.)

\begin{figure}[h]
	\centering
	\includegraphics[scale=1]{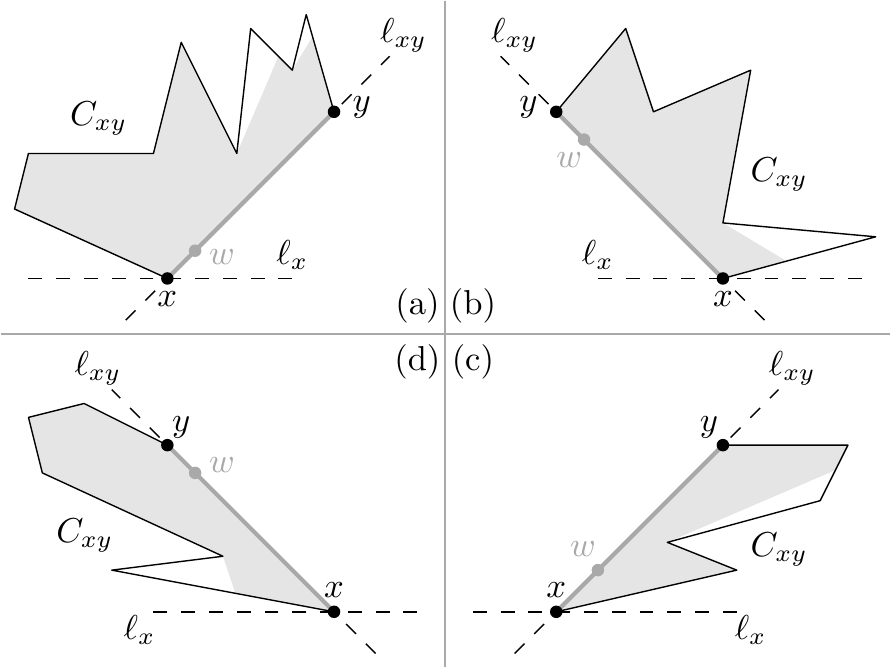}
	\caption{\small $\xy$ is a constructed edge, and the gray area is the set of all points of $\vis(w)$ that lie below/above $\ell_{xy}$. In (a) and (b) the pocket $C_{xy}$ lies above $\ell_{xy}$, and thus $\xy$ is a lower edge. In (c) and (d) the pocket $C_{xy}$ lies below $\ell_{xy}$, and $\xy$ is an upper edge. In (a) and (c) the slope of $\xy$ is positive, and in (b) and (d) it is negative.}
	\label{fig:onePocketLemma}
\end{figure}

Let $\xy$ be a constructed edge such that $x$ is below $y$ (w.r.t. the $y$-axis), and denote by $C_{xy}$ the pocket adjacent to $\xy$.
We say that $C_{xy}$ lies above (resp. below) $\ell_{xy}$, if for any point $w$ in the interior of $\xy$, 
all the points that $w$ sees (points of $\vis(w)$) that lie above (resp. below) $\ell_{xy}$, belong to $C_{xy}$ (see \Cref{fig:onePocketLemma}). Notice that since $\xy$ divides $P$ into two parts, $C_{xy}$ and $P\setminus C_{xy}$, `stepping off' $\xy$ to one side, places us in the interior of $C_{xy}$ (and `stepping off' $\xy$ to the other side places us in the interior of $P\setminus C_{xy}$).

\begin{figure}[h]
	\centering
	\includegraphics[scale=1]{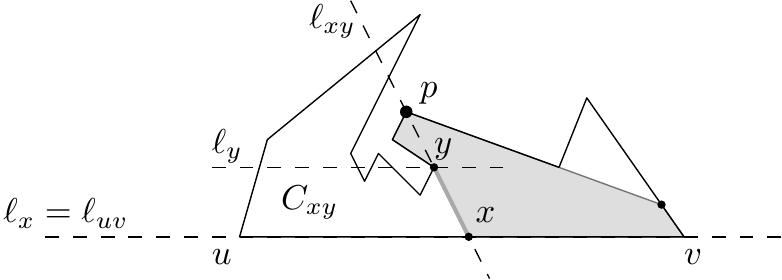}
	\caption{\small The gray area is $\vis(p)$. $\xy$ is an upper edge (i.e., $C_{xy}$ lies below $\ell_{xy}$) and $x$ is on $\uv$, but some points of $C_{xy}$ lie above $\ell_{xy}$.}
	\label{fig:onePocketLemmaUV}
\end{figure}

Note that if $C_{xy}$ lies above (resp. below) $\ell_{xy}$, it does not necessarily mean that all the points of $C_{xy}$ lie above (resp. below) $\ell_{xy}$. Indeed, when $x\in\uv$, $C_{xy}$ may have points on both sides of $\ell_{xy}$ (see \Cref{fig:onePocketLemmaUV}).

We say that $\xy$ is an \dfn{upper edge} (resp. \dfn{lower edge}), if $C_{xy}$ lies below (resp. above) $\ell_{xy}$.
Notice that by \Cref{clm:singleSegmentOfUV}, $\vis(p)$ has at most two constructed edges with an endpoint in $\uv$. Moreover, at most one of these edges is an upper edge.

\begin{figure}[h]
	\centering
	\includegraphics[scale=1]{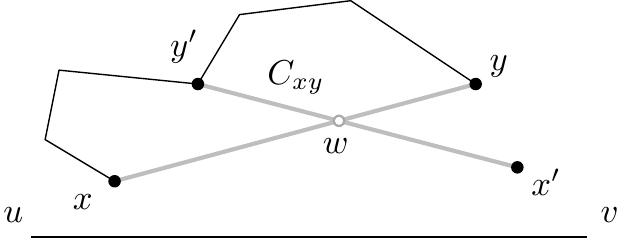}
	\caption{Two constructed edges that cross each other.}
	\label{fig:twoPocketsLemma}
\end{figure}

\begin{claim}\label{clm:two_pockets}
	Let $\xy$ and $\xyp$ be two constructed edges (which belong to two different visibility polygons), such that $\xy$ crosses $\xyp$, and $y'$ is on the same side of $\ell_{xy}$ as $C_{xy}$. Then, $y'\in \partial C_{xy}$ (see \Cref{fig:twoPocketsLemma}).
\end{claim}
\begin{proof}
	Let $w$ be the crossing point of $\xy$ and $\xyp$. Then, $w$ is a point in the interior of $\xy$ that sees $y'$, and $y'$ is on the same side of $\ell_{xy}$ as $C_{xy}$, so $y'\in C_{xy}$. Since $\xyp$ is a constructed edge, $y'\in\partial P$ and thus  $y'\in \partial C_{xy}$.
\end{proof}

\begin{claim}\label{clm:pocket}
	Let $\xy$ be a constructed edge of $\vis(p)$ such that $x$ is below $y$, then $C_{xy}$ lies above $\ell_x$ (in the weak sense), i.e., every point in $C_{xy}$ is either above or on $\ell_x$. Moreover, if $x\notin \uv$, then $C_{xy}$ lies strictly above $\ell_x$, i.e., every point in $C_{xy}$ is above $\ell_x$.
\end{claim}
\begin{proof}
	If $x\in \uv$, then since $P$ is above $\ell_{uv}=\ell_x$, we get that $C_{xy}$ is above $\ell_x$.
		
	Observe that if $\partial C_{xy} \cap \uv \ne \emptyset$, then $x\in \uv$. Otherwise, it follows that $\uv\subseteq \partial C_{xy}$, and since the entire pocket $C_{xy}$ is not visible from $p$, we get that $p$ is not visible from $\uv$, which contradicts the fact that $P$ is a WV-polygon.
		
	If $x\notin \uv$, then $x$ is above $\ell_{uv}$ (see \Cref{fig:onePocketLemma}). By the observation above, $\partial C_{xy} \cap \uv = \emptyset$ (and thus of course $C_{xy} \cap \uv = \emptyset$). Let $z$ be any point in $C_{xy}$. Since $P$ is a WV-polygon, there exists a point $z'$ on $\uv$ such that $\seg{zz'}\subseteq P$. The segment $\seg{zz'}$ has to cross $\xy$, because $z\in C_{xy}$ and $z'\notin C_{xy}$. Since $x$ is below $y$, this crossing point is above $x$, which in turn is above $z'$, so we get that $z$ is above $x$. We conclude that $C_{xy}$ lies strictly above $\ell_x$.
\end{proof}

\subsection{Holes} 
Let $G \subseteq P$ be a set of points that guards $\partial P$. A connected subset $H$ of $P$ is 
a \dfn{hole} in $P$ w.r.t. $G$ if (i) there is no point in $H$ that is visible from $G$, and (ii) $H$ is maximal in the sense that any connected subset of $P$ that strictly contains $H$ has a point that is visible from $G$. 

Let $H$ be a hole in $P$ w.r.t. $G$, then clearly $H$ is a simple polygon. Each edge of $H$ lies on some constructed edge $e$, and we say that $H$ (and this edge of $H$) \dfn{lean} on the edge $e$. Notice that $H$ is fully contained in the pocket adjacent to $e$. Moreover, since $H \cap \partial P = \emptyset$, we can view $H$ as an intersection of halfplanes (defined by the lines containing the constructed edges on which the edges of $H$ lean), and thus obtain the following observation.

\begin{observation}
	Any hole $H$ in $P$ w.r.t. $G$ is a convex polygon.
\end{observation}

Another immediate but useful observation is the following.
\begin{observation}\label{obs:upperlower}
	Any hole $H$ in $P$ w.r.t. $G$ leans on at least one upper edge and at least one lower edge.
\end{observation}

Next, we show that any hole can be guarded by a single vertex, and that such a vertex can be found in polynomial time. In the following lemma we prove a slightly more general claim. 
\begin{lemma}
	\label{lem:convexHoleGuard}
	Let $S$ be a simple convex polygonal region contained in $P$. Then, 
	there exists a vertex $w$ of $P$, such that $S$ is visible from $w$, i.e., every point in $S$ is visible from $w$. Moreover, $w$ can be found in polynomial time.
\end{lemma}

\begin{proof}
\begin{figure}[H]
		\centering
		\includegraphics[scale=1]{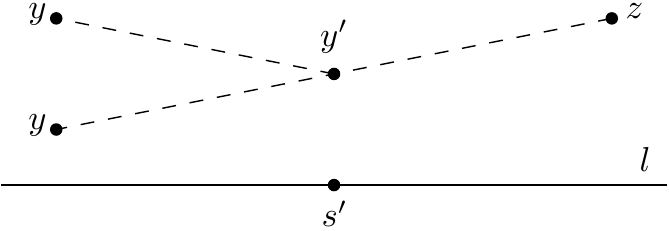}\hspace{2cm}
		\includegraphics[scale=1]{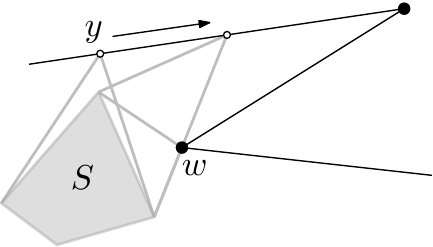}
		\caption{Proof of \Cref{lem:convexHoleGuard}.}
		\label{fig:closestPointToS}
	\end{figure}
If $P$ and $S$ share a vertex, then we are done.
Otherwise, 
let $y$ be a point on $\partial P$ that is closest to $S$, where the distance between a point $p$ and $S$ is $dist(p,S) = \min\{||p-s||\,|\,s \in S\}$, and let $s$ be a point in $S$ closest to $y$.
We first prove that $S$ is visible from $y$. If $y$ is also on $\partial S$, then $y$ clearly sees $S$, so assume that $y$ is not on $\partial S$ and that there exists a point $t \in S$ that is not visible from $y$. Then, there exists a point $y'$ on $\partial P$ that lies in the interior of the segment $\overline{yt}$. Let $s'$ be a point in $S$ closest to $y'$, and let $l$ be the line through $s'$ and perpendicular to $\overline{y's'}$. Assume w.l.o.g. that $l$ is horizontal and that $S$ is contained in the bottom halfplane defined by $l$ (see \Cref{fig:closestPointToS}, left). Then $y'$ is above $l$ and its projection onto $l$ is $s'$. Now, it is impossible that $y$ is not above $y'$ (in terms of $y$-coordinate), since this would imply that $t$ is above $l$. So $y$ must be above $y'$, but then $||y-s|| > ||y'-s'||$ and we have reached a contradiction.

Now, if $y$ is a vertex of $P$, then we are done. Otherwise, we slide $y$ along $\partial P$, in any one of the directions, until $CH(S \cup y)$ meets a vertex $w$ of $P$ (see \Cref{fig:closestPointToS}, right). The vertex $w$ is either an endpoint of the edge of $P$ on which we slide $y$, or it is another vertex of $P$ that lies on one of the tangents to $S$ through $y$. In both cases, $w$ clearly sees $S$. 
\end{proof}

\section{The Algorithm}
\label{sec:algorithm}
We show that given a set $G$ of vertices that guards $\partial P$, one can find a set of vertices $G'$ of size at most $|G|$ such that $G\cup G'$ guards $P$ (boundary plus interior). Let $E$ be the set of the constructed edges of the visibility polygons of the vertices in $G$.

\begin{algorithm}[H]
	\vspace{10pt}
	For each upper edge $e=\xy$ in $E$ with $x\in \uv$:
	\begin{enumerate}
	\item Find the topmost intersection point $p$ of $e$ with an edge $e'=\xyp$ in $E$ with $x'\in \uv$ such that $x'$ is on the same side of $\ell_{xy}$ as the pocket $C_{xy}$.
	\item If such a point $p$ exists (then the triangle $\triangle xpx'$ is contained in $P$), find a vertex that guards $\triangle xpx'$
			(see description in the proof of \Cref{lem:convexHoleGuard}) and add it to $G'$. 
	\end{enumerate}
	\caption{\label{alg:main} Guarding the interior of $P$}
\end{algorithm}

Now our goal is to show that for any hole $H$ in $P$ w.r.t. $G$ there exists a vertex in $G'$ that guards it. More precisely, we show that $H$ is contained in one of the triangles considered by the algorithm.

Let $H$ be a hole with vertices $h_1,\dots,h_k$. Using the following lemma, we first show that $H$ leans on an upper edge with an endpoint in $\uv$. By \Cref{obs:upperlower}, $H$ leans on at least one upper edge.

\begin{figure}[h]
	\centering
	\includegraphics[scale=1]{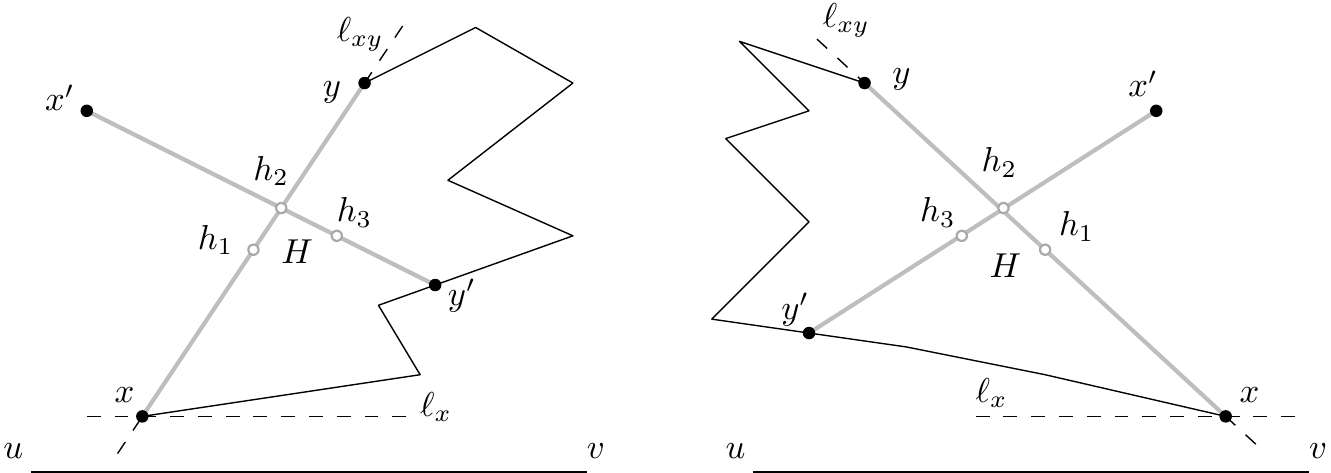}
	\caption{\small Two edges of the hole $H$: $h_1h_2$ and $h_2h_3$. The edge $h_1h_2$ leans on an upper edge and $h_1$ is below $h_2$, thus $h_2h_3$ has to lean on an upper edge such that $h_2$ is below $h_3$ (and not as drawn in the figure).}
	\label{fig:upperEdgeLemma}
\end{figure}

\begin{lemma}\label{lem:upperEdge}
	Assume that $h_1h_2$ leans on an upper edge $\xy$, such that $x$ is below $y$, $x \notin \uv$, and $h_1$ is below $h_2$. Then $h_2h_3$ also leans on an upper edge and $h_2$ is below $h_3$.
\end{lemma}
\begin{proof}
	Since $\xy$ is an upper edge, $H$ lies below $\ell_{xy}$, and by \Cref{clm:pocket}, $H$ lies strictly above $\ell_x$ (see \Cref{fig:upperEdgeLemma}).
	
	Let $\xyp$ be the constructed edge on which $h_2h_3$ is leaning, and assume that $y'$ is on the same side of $\ell_{xy}$ as $h_3$. ($h_2$ is the crossing point of $\xy$ and $\xyp$, so $x'$ and $y'$ are on different sides of $\ell_{xy}$.) Since $h_3$ is a vertex of the hole $H$, both $h_3$ and $y'$ are on the same side of $\ell_{xy}$ as $C_{xy}$. Moreover, by \Cref{clm:two_pockets} we get that $y' \in \partial C_{xy}$, and therefore $y'$ is strictly above $x$.
	
	Next we show that $h_2$ is below $h_3$, by showing that $x'$ is below $y'$. Indeed, if $x'$ is above $y'$, then since $h_1$ is a vertex of the hole $H$, both $h_1$ and $x$ are on the same side of $\ell_{x'y'}$ as $C_{x'y'}$. Again by \Cref{clm:two_pockets} we get that $x \in \partial C_{x'y'}$, and therefore $x$ is strictly above $y'$, a contradiction.
	
	Finally, $\xyp$ is clearly an upper edge, since $h_1$ (which is in $C_{x'y'}$) is below $\xyp$.
\end{proof}

\begin{figure}[h]
	\centering
	\includegraphics[page=1,scale=1.2]{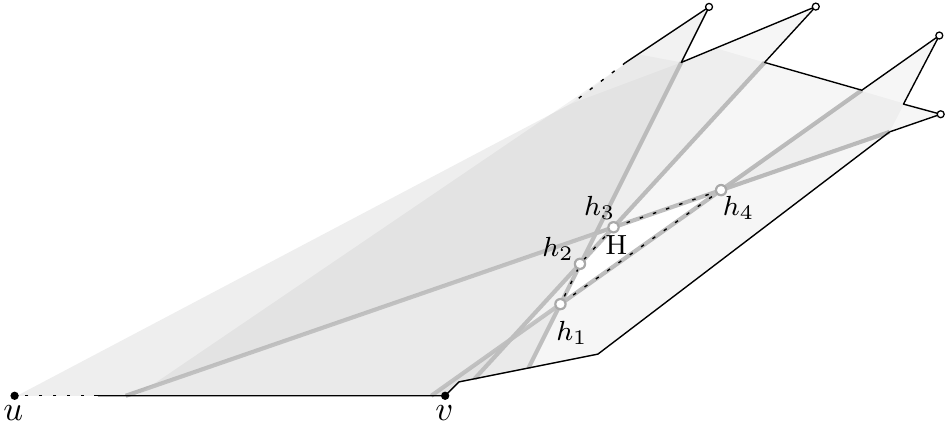}
	\caption{\small The edges of $H$ lean on a sequence of upper edges.}
	\label{fig:upperEdges}
\end{figure}

This implies that every hole $H$ leans on an upper edge with an endpoint in $\uv$: 
$H$ has at least one edge $h_1h_2$ that leans on an upper edge $\seg{x_1y_1}$, and $h_1$ is below $h_2$. If $x_1$ or $y_1$ are on $\uv$ then we are done. Otherwise, by \Cref{lem:upperEdge}, the edge $h_2h_3$ also leans on an upper edge $\seg{x_2y_2}$, and $h_2$ is below $h_3$. Again, if one of $x_2,y_2$ is on $\uv$ then we are done, otherwise by \Cref{lem:upperEdge} $h_3h_4$ leans on an upper edge and $h_3$ is below $h_4$. The process must end before it reaches $h_{k-1}h_k$, since, if $h_{k-1}h_k$ leans on an upper edge $\seg{x_{k-1}y_{k-1}}$ such that none of $x_{k-1},y_{k-1}$ is on $\uv$, then by applying \Cref{lem:upperEdge} once again we get that $h_k$ is below $h_1$, which is impossible (see \Cref{fig:upperEdges}).
But the only way for the process to end is when it reaches an upper edge $\seg{x_ty_t}$ such that $x_t$ or $y_t$ are on $\uv$.

\begin{claim}\label{lem:touchUV}
	Assume that $h_1h_2$ leans on an upper edge $\xy$ such that $x$ is below $y$ and $x\in\uv$, and that $h_1$ is below $h_2$. Then $h_2h_3$ leans on a constructed edge $\xyp$ (where $x'$ is below $y'$) such that $x' \in \uv$, 
\end{claim}
\begin{proof}	
	Let $\xyp$ be the constructed edge on which $h_2h_3$ is leaning such that $x'$ is below $y'$. The constructed edges $\xy$ and $\xyp$ cross each other at the point $h_2$, so by \Cref{clm:two_pockets} we get that $x \in \partial C_{x'y'}$, because $h_1$ and $x$ are on the same side of $\ell_{x'y'}$. Now, by \Cref{clm:pocket}, $x$ lies (weakly) above $\ell_{x'}$, but $x\in\uv$, so $x' \in \ell_{uv}$.
\end{proof}

We are now ready to prove the correctness of our algorithm.

\begin{claim}
	Any hole $H$ in $P$ w.r.t. $G$ is contained in one of the triangles considered by \Cref{alg:main}. 
\end{claim}
\begin{proof}
	By the argument immediately following \Cref{lem:upperEdge}, $H$ has an edge which leans on an upper edge with an endpoint on $\uv$. Let $h_1h_2$, where $h_1$ is below $h_2$, be such an edge of $H$ of minimum (absolute) slope, and let $\xy$, where $x\in\uv$, be the upper edge on which $h_1h_2$ is leaning. By \Cref{lem:touchUV}, the edge $h_2h_3$ leans on a constructed edge $\xyp$ such that $x'\in\uv$. 
		
	Assume first that the algorithm found a point $p$ on $\xy$. If $p$ is not below $h_2$, then the triangle corresponding to $p$ contains $H$ and we are done.
	If $p$ is below $h_1$ (or $p=h_1$), then $x'$ is not on the same side of $\xy$ as $C_{xy}$ (because otherwise, $p$ would be a point not below $h_2$). But, if so, then the (absolute) slope of $h_2h_3$ is smaller than that of $h_1h_2$, a contradiction. 
	Now, if the algorithm did not find such a point $p$ on $\xy$, then as before this means that $x'$ is not on the same side of $\xy$ as $C_{xy}$ and we reach a contradiction.   
	
\end{proof}

\begin{figure}[h]
	\centering
	\includegraphics[scale=1]{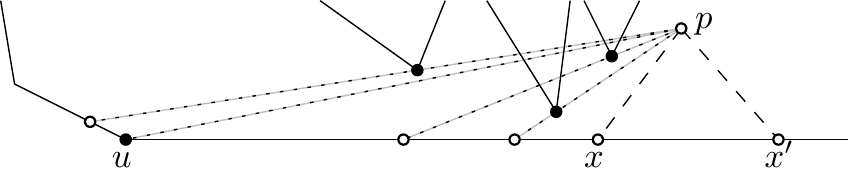}
	\caption{\small Finding a vertex of $P$ that guards the triangle $\triangle xpx'$.}
	\label{fig:linear_time}
\end{figure}

\paragraph{Running time.}
The running time of the algorithm of~\cite{Katz18} for finding a set of vertices $G$ of size $O(1+\eps/2)\opt$ that guards $\partial P$ is $O(n^{O(1/\eps^2)})$, for any $\eps > 0$. The set of constructed edges $E$ can be computed in $O(|G|n)$ time~\cite{JoeS87,Lee83}. Notice that \Cref{alg:main} only uses a subset $E'$ of the constructed edges, namely, those with an endpoint in $\uv$, and by \Cref{clm:singleSegmentOfUV} we have that $|E'| \le 2|G|$. Therefore, the total number of intersection points (Step~1 of \Cref{alg:main}) is $O(|G|^2)$. Moreover, a vertex that guards the triangle $\triangle xpx'$ (Step~2 of \Cref{alg:main}) can be found in linear time by the following simple algorithm. (We could use the algorithm described in the proof of \Cref{lem:convexHoleGuard}, but in this special case it is not necessary.) Assume, w.l.o.g., that $x$ is to the left of $x'$. Then, for each vertex $w$ of $P$ such that $w$ is below $p$, compute the crossing point (if it exists) between $\uv$ and the ray emanating from $p$ and passing through $w$. Now, among the vertices whose corresponding crossing point is between $u$ and $x$, pick the closest one to $x$ (see \Cref{fig:linear_time}). 
Thus, the total running time of \Cref{alg:main} (given the set $G$) is $O(|G|n)=O(n^2)$. Note that for any $\eps < \frac{1}{\sqrt{2}}$, the total running time is dominated by $O(n^{O(1/\eps^2)})$.

\begin{figure}[h]
	\centering
	\includegraphics[scale=1]{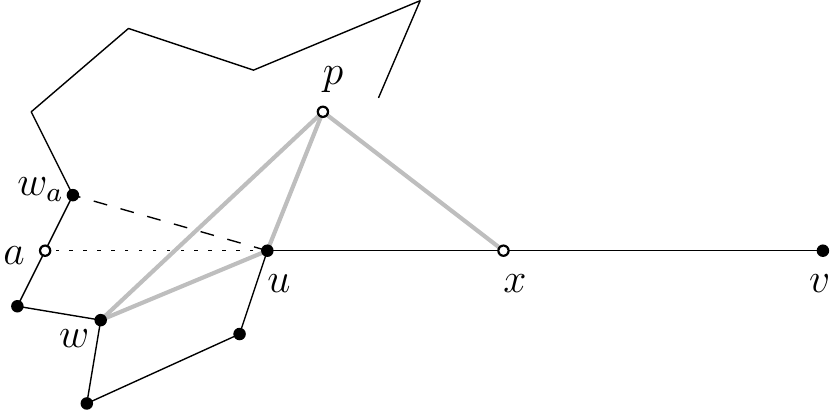}
	\caption{\small A polygon with a concave angle at $u$.}
	\label{fig:convexity}
\end{figure}

\paragraph{Removing the convexity assumption.} Up to now, we have assumed that the angles at $u$ and at $v$ are convex. As in~\cite{Katz18}, this assumption can be easily removed. Assume, e.g., that the angle at $u$ is concave, and let $a$ be the first point on $P$'s boundary (moving clockwise from $u$) that lies on the $x$-axis (see \Cref{fig:convexity}).
Then, every point in the open portion of the boundary between $u$ and $a$ is visible from $u$ and is not visible from any other point on $\uv$. Moreover, for any vertex $w$ in this portion of $P$'s boundary, if $w$ sees some point $p$ in $P$ above the $x$-axis, then so does $u$. Indeed, since $P$ is weakly-visible from $\uv$, there exists a point $x\in\uv$ that sees $p$. In other words, $\seg{xp}$ is contained in $P$, as well as $\seg{uw}$ and $\seg{wp}$. Thus, the quadrilateral $uwpx$ does not contain points of $\partial P$ in its interior, and since $\seg{up}$ is contained in it, we conclude that $u$ sees $p$. 
Therefore, we may assume that an optimal guarding set does not include a vertex from this portion. Moreover, we may assume that the size of an optimal guarding set is greater than some appropriate constant, since otherwise we can find such a set in polynomial time. Now, let $w_a$ be the first vertex following $a$. We place a guard at $u$ and replace the portion of $P$'s boundary between $u$ and $w_a$ by the edge $\seg{uw_a}$. Notice that every point in the removed region is visible from $u$. Similarly, if the angle at $v$ is concave, we define the point $b$ and the vertex $w_b$ (by moving counterclockwise from $v$), place a guard at $v$, and replace the portion of $P$'s boundary between $v$ and $w_b$ by the edge $\seg{vw_b}$. Finally, we apply the algorithm of~\cite{Katz18} to the resulting polygon, after adjusting its parameters so that together with $u$ and $v$ we still get a $(1+\eps/2)$-approximation of an optimal guarding set for $\partial P$.

\begin{theorem}
	Given a WV-polygon $P$ with $n$ vertices and $\eps > 0$, one can find in $O(n^{\max\{2,O(1/\eps^2)\}})$ time a subset $G$ of the vertices of $P$, such that $G$ guards $P$ (boundary plus interior) and $G$ is of size at most $(2+\eps)\opt$, where $\opt$ is the size of a minimum-cardinality such set.
\end{theorem}

\paragraph{Boundary guards.}
Let $\opt^B$ be the size of a minimum-cardinality set of points on $P$'s boundary (except the interior of the edge $\uv$) that guards $P$ (boundary plus interior), and let $\opt^B_\partial$ be the size of a minimum-cardinality such set that guards $P$'s boundary; clearly, $\opt^B_\partial \le \opt^B \le \opt$. A PTAS for finding a set $G^B$ of points on $(\partial P \setminus \uv) \cup \{u,v\}$ that guards $\partial P$ is described in~\cite{Katz18}, that is, $|G^B| \le (1+\eps)\opt^B_\partial$, for any $\eps > 0$. Its running time is $O(n^{O(1/\eps^2)})$, and it is similar to the corresponding PTAS of Friedrichs et al.~\cite{FHKS16} for the case of 1.5D-terrains. Given the set $G^B$ as input, we can apply our algorithm as is and obtain a set $G'$ of boundary points of size at most $|G^B|$ such that $G^B \cup G'$ guards $P$. Thus, we have $|G^B|+|G'|\le 2 |G^B| \le (2+\eps) \opt^B$.

\begin{corollary}
Given a WV-polygon $P$ (w.r.t to edge $\uv$) with $n$ vertices and $\eps > 0$, one can find in $O(n^{\max\{2,O(1/\eps^2)\}})$ time a set $G$ of points on $(\partial P\setminus \uv) \cup \{u,v\}$, such that $G$ guards $P$ (boundary plus interior) and $G$ is of size at most $(2+\eps)\opt^B$, where $\opt^B$ is the size of a minimum-cardinality such set.
\end{corollary}

\section{Polygons weakly visible from a chord}
\label{sec:chord}
A \emph{chord} in a simple polygon $P$ is a line segment whose endpoints are on the boundary of $P$ and whose interior is contained in the interior of $P$. In particular, any diagonal of $P$ is a chord in $P$. 

In this section, we show that our method can be extended to the case where $P$ is weakly visible from a chord $\uv$ (in $P$), i.e., every point in $P$ is visible from some point on $\uv$. 
More precisely, we show that given a set $G$ of vertices that guards the boundary of such a polygon $P$, one can find a set $G'$ of size at most $2|G|$ such that $I = G \cup G'$ guards $P$ (boundary plus interior).
Thus, given a $c$-approximation algorithm for guarding the boundary of a polygon $P$ weakly visible from a chord, we provide a $3c$-approximation algorithm for guarding $P$.

The chord $\uv$ slices $P$ into two (sub)polygons, where each of them is a WV-polygon w.r.t. the edge $uv$ (see \Cref{fig:visiblityPolygonChord}). Thus, we would like to apply \Cref{alg:main} to each of these polygons separately.

\begin{figure}[h]
	\centering
	\includegraphics[scale=1]{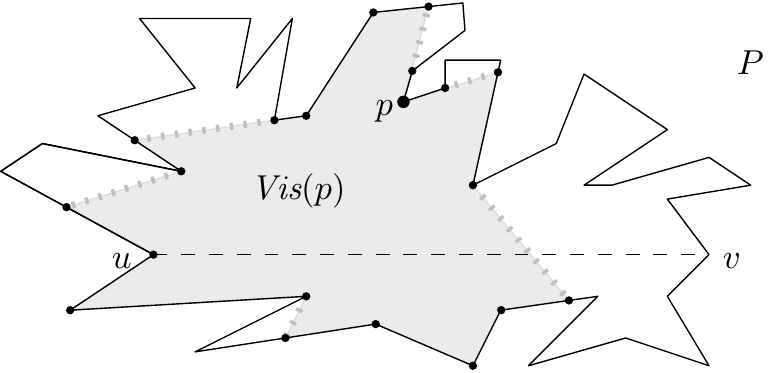}
	\caption{The visibility polygon of a point $p\in P$.}
	\label{fig:visiblityPolygonChord}
\end{figure}

The definitions of visibility polygons, pockets, and holes, apply with no change to polygons weakly visible from a chord. However, in order for the rest of our claims to be correct, we have to update \Cref{clm:singleSegmentOfUV} and \Cref{clm:pocket}.
First, we replace \Cref{clm:singleSegmentOfUV} by the following claim; the proof remains unchanged.
\begin{claim}\label{clm:singleSegmentOfUVChord}
	For any point $p\in P$, $\vis(p) \cap \uv$ is connected.
\end{claim}
Notice that by \Cref{clm:singleSegmentOfUVChord}, $\vis(p)$ has at most two constructed edges that \textbf{cross} $\uv$, and again at most one of these edges is an upper edge.
Next, we replace \Cref{clm:pocket} by the following two claims.
\begin{claim}\label{clm:pocket2}
	Let $\xy$ be a constructed edge of $\vis(p)$, such that $x$ is below $y$ and $\xy$ is strictly above $\uv$, then $C_{xy}$ lies strictly above $\ell_x$, i.e., every point in $C_{xy}$ is above $\ell_x$.
\end{claim}
\begin{proof}
	First, observe that $C_{xy} \cap \uv = \emptyset$, otherwise, it follows that $\uv\subseteq C_{xy}$, and since the entire pocket $C_{xy}$ is not visible from $p$, we get that $p$ is not visible from $\uv$, which contradicts the fact that $P$ is a WV-polygon.
	
	Let $z$ be any point in $C_{xy}$. Since $P$ is a WV-polygon, there exists a point $z'$ on $\uv$ such that $\seg{zz'}\subseteq P$. The segment $\seg{zz'}$ has to cross $\xy$, because $z\in C_{xy}$ and $z'\notin C_{xy}$. Since $x$ is below $y$, this crossing point is above $x$, which in turn is above $z'$, so we get that $z$ is above $x$. We conclude that $C_{xy}$ lies strictly above $\ell_x$.
\end{proof}

\begin{claim}\label{clm:pocket3}
	Let $\xy$ be a constructed edge of $\vis(p)$, such that $x$ is below $y$ and $\xy$ crosses $\uv$ at a point $w$, then $\ell_{uv}$ intersects $\partial C_{xy}$ at a single point which is either $u$ or $v$.
\end{claim}
\begin{proof}
	$\uv$ intersects $\partial C_{xy}$ at a single point which is either $u$ or $v$, since $\partial C_{xy} \subseteq \partial P$ and $\uv$ is a chord. Now, if $\ell_{uv}$ intersects $\partial P$ at a point different than $u,v$, then this point is not visible from $\uv$.
\end{proof}

Denote by $P_1$ the polygon above $\uv$ and by $P_2$ the polygon below $\uv$.
As in the previous section, let $E$ be the set of the constructed edges of the visibility polygons of the vertices in $G$. 
We run \Cref{alg:main} on $P_1$ and $P_2$ separately, to find a set of vertices $G'$ (of $P_1$) such that $G\cup G'$ guards $P_1$ (boundary plus interior), and set of vertices $G''$ (of $P_2$) such that $G\cup G''$ guards $P_2$. Both $G'$ and $G''$ are of size at most $|G|$, and thus $G\cup G'\cup G''$ is a set of vertices of size at most $3|G|$ that guards $P$.

When running the algorithm on $P_1$, we ignore the existence of $P_2$. That is, if an edge $\xy\in E$ crosses $\uv$, then we only consider its part that is contained in $P_1$ (and the part of its associated pocket that is contained in $P_1$). Similarly, if a hole $H$ lies on both sides of $\uv$, we only consider its part that is contained in $P_1$. (Notice that the role of an edge $\xy\in E$ that crosses $\uv$ switches from an upper edge in $P_1$ to a lower edge in $P_2$, or vice versa. So, each such edge is an upper edge in either $P_1$ or $P_2$.) Moreover, in this case we get a hole in $P_1$ that has an edge which is not a constructed edge of a visibility polygon. But this does not cause a problem, since such holes must also have upper edges, and furthermore, due to \Cref{clm:pocket2} we get that each of these upper edges has an endpoint on $\overline{uv}$.
Finally, notice that after computing $E$, we do not use the set $G$ in our algorithm and proofs; it only reappears as one of the three sets whose union is the final guarding set.

\begin{theorem}
	Given a polygon $P$ with $n$ vertices that is weakly visible from a chord, and a set $G$ of vertices of $P$ (or a set $G$ of points on $\partial P$) that guards $\partial P$, one can find in polynomial time a set $G'$ of size at most $2|G|$ such that $G \cup G'$ guards $P$ (boundary plus interior). If $G$ is a $c$-approximation for guarding $\partial P$, then $G \cup G'$ is a $3c$-approximation for guarding the entire polygon $P$.
\end{theorem}

\bibliography{refs}

\begin{thebibliography}{10}

\bibitem{AAM18}
Mikkel Abrahamsen, Anna Adamaszek, and Tillmann Miltzow.
\newblock The art gallery problem is {\(\exists\)} {{\(\mathbb{R}\)}}-complete.
\newblock In {\em Proceedings of the 50th Annual {ACM} {SIGACT} Symposium on
  Theory of Computing, {STOC} 2018, Los Angeles, CA, USA, June 25-29, 2018},
  pages 65--73, 2018.
\newblock URL: \url{https://doi.org/10.1145/3188745.3188868}, \href
  {http://dx.doi.org/10.1145/3188745.3188868}
  {\path{doi:10.1145/3188745.3188868}}.

\bibitem{Aggarwal84}
Alok Aggarwal.
\newblock The art gallery theorem: its variations, applications and algorithmic
  aspects.
\newblock 1984.

\bibitem{AT81}
David Avis and Godfried~T. Toussaint.
\newblock An optimal algorithm for determining the visibility of a polygon from
  an edge.
\newblock {\em {IEEE} Trans. Computers}, 30(12):910--914, 1981.
\newblock URL: \url{https://doi.org/10.1109/TC.1981.1675729}, \href
  {http://dx.doi.org/10.1109/TC.1981.1675729}
  {\path{doi:10.1109/TC.1981.1675729}}.

\bibitem{bhattacharya2017constant}
Pritam Bhattacharya, Subir~Kumar Ghosh, and Sudebkumar Pal.
\newblock Constant approximation algorithms for guarding simple polygons using
  vertex guards.
\newblock {\em arXiv:1712.05492}, 2017.

\bibitem{BGPR18}
Pritam Bhattacharya, Subir~Kumar Ghosh, Sudebkumar Pal, and Bodhayan Roy.
\newblock Personal communication, 2018.

\bibitem{BGR17}
Pritam Bhattacharya, Subir~Kumar Ghosh, and Bodhayan Roy.
\newblock Approximability of guarding weak visibility polygons.
\newblock {\em Discrete Applied Mathematics}, 228:109--129, 2017.
\newblock URL: \url{https://doi.org/10.1016/j.dam.2016.12.015}, \href
  {http://dx.doi.org/10.1016/j.dam.2016.12.015}
  {\path{doi:10.1016/j.dam.2016.12.015}}.

\bibitem{BM17}
{\'{E}}douard Bonnet and Tillmann Miltzow.
\newblock An approximation algorithm for the art gallery problem.
\newblock In {\em 33rd International Symposium on Computational Geometry, SoCG
  2017, July 4-7, 2017, Brisbane, Australia}, pages 20:1--20:15, 2017.
\newblock URL: \url{https://doi.org/10.4230/LIPIcs.SoCG.2017.20}, \href
  {http://dx.doi.org/10.4230/LIPIcs.SoCG.2017.20}
  {\path{doi:10.4230/LIPIcs.SoCG.2017.20}}.

\bibitem{DHN94}
Gautam Das, Paul~J. Heffernan, and Giri Narasimhan.
\newblock Finding all weakly-visible chords of a polygon in linear time.
\newblock {\em Nord. J. Comput.}, 1(4):433--457, 1994.

\bibitem{DKDS07}
Ajay Deshpande, Taejung Kim, Erik~D. Demaine, and Sanjay~E. Sarma.
\newblock A pseudopolynomial time {$O(\log n)$}-approximation algorithm for art
  gallery problems.
\newblock In {\em Algorithms and Data Structures, 10th International Workshop,
  {WADS} 2007, Halifax, Canada, August 15-17, 2007, Proceedings}, pages
  163--174, 2007.
\newblock URL: \url{https://doi.org/10.1007/978-3-540-73951-7\_15}, \href
  {http://dx.doi.org/10.1007/978-3-540-73951-7\_15}
  {\path{doi:10.1007/978-3-540-73951-7\_15}}.

\bibitem{EH06}
Alon Efrat and Sariel Har{-}Peled.
\newblock Guarding galleries and terrains.
\newblock {\em Inf. Process. Lett.}, 100(6):238--245, 2006.
\newblock URL: \url{https://doi.org/10.1016/j.ipl.2006.05.014}, \href
  {http://dx.doi.org/10.1016/j.ipl.2006.05.014}
  {\path{doi:10.1016/j.ipl.2006.05.014}}.

\bibitem{EidenbenzSW01}
Stephan Eidenbenz, Christoph Stamm, and Peter Widmayer.
\newblock Inapproximability results for guarding polygons and terrains.
\newblock {\em Algorithmica}, 31(1):79--113, 2001.
\newblock URL: \url{https://doi.org/10.1007/s00453-001-0040-8}, \href
  {http://dx.doi.org/10.1007/s00453-001-0040-8}
  {\path{doi:10.1007/s00453-001-0040-8}}.

\bibitem{FHKS16}
Stephan Friedrichs, Michael Hemmer, James King, and Christiane Schmidt.
\newblock The continuous 1.5{D} terrain guarding problem: Discretization,
  optimal solutions, and {PTAS}.
\newblock {\em JoCG}, 7(1):256--284, 2016.
\newblock URL: \url{https://doi.org/10.20382/jocg.v7i1a13}, \href
  {http://dx.doi.org/10.20382/jocg.v7i1a13} {\path{doi:10.20382/jocg.v7i1a13}}.

\bibitem{Ghosh87}
Subir~Kumar Ghosh.
\newblock Approximation algorithms for art gallery problems.
\newblock pages 429--434, 1997.

\bibitem{Ghosh10}
Subir~Kumar Ghosh.
\newblock Approximation algorithms for art gallery problems in polygons.
\newblock {\em Discrete Applied Mathematics}, 158(6):718--722, 2010.
\newblock URL: \url{https://doi.org/10.1016/j.dam.2009.12.004}, \href
  {http://dx.doi.org/10.1016/j.dam.2009.12.004}
  {\path{doi:10.1016/j.dam.2009.12.004}}.

\bibitem{GMPSM93}
Subir~Kumar Ghosh, Anil Maheshwari, Sudebkumar~Prasant Pal, Sanjeev Saluja, and
  C.~E.~Veni Madhavan.
\newblock Characterizing and recognizing weak visibility polygons.
\newblock {\em Comput. Geom.}, 3:213--233, 1993.
\newblock URL: \url{https://doi.org/10.1016/0925-7721(93)90010-4}, \href
  {http://dx.doi.org/10.1016/0925-7721(93)90010-4}
  {\path{doi:10.1016/0925-7721(93)90010-4}}.

\bibitem{Gibson14}
Matt Gibson, Gaurav Kanade, Erik Krohn, and Kasturi Varadarajan.
\newblock Guarding terrains via local search.
\newblock {\em Journal of Computational Geometry}, 5(1):168--178, 2014.

\bibitem{JoeS87}
B.~Joe and R.~B. Simpson.
\newblock Corrections to {L}ee's visibility polygon algorithm.
\newblock {\em {BIT}}, 27(4):458--473, 1987.

\bibitem{Katz18}
Matthew~J. Katz.
\newblock A {PTAS} for vertex guarding weakly-visible polygons --- an extended
  abstract.
\newblock {\em CoRR}, abs/1803.02160, 2018.
\newblock URL: \url{http://arxiv.org/abs/1803.02160}, \href
  {http://arxiv.org/abs/1803.02160} {\path{arXiv:1803.02160}}.

\bibitem{Ke87}
Yan Ke.
\newblock Detecting the weak visibility of a simple polygon and related
  problems.
\newblock In {\em Technical report}. Johns Hopkins University, 1987.

\bibitem{KK11}
James King and David~G. Kirkpatrick.
\newblock Improved approximation for guarding simple galleries from the
  perimeter.
\newblock {\em Discrete {\&} Computational Geometry}, 46(2):252--269, 2011.
\newblock URL: \url{https://doi.org/10.1007/s00454-011-9352-x}, \href
  {http://dx.doi.org/10.1007/s00454-011-9352-x}
  {\path{doi:10.1007/s00454-011-9352-x}}.

\bibitem{KN13}
Erik Krohn and Bengt Nilsson.
\newblock Approximate guarding of monotone and rectilinear polygons.
\newblock {\em Algorithmica}, 66:564--594, 07 2013.
\newblock \href {http://dx.doi.org/10.1007/s00453-012-9653-3}
  {\path{doi:10.1007/s00453-012-9653-3}}.

\bibitem{LL86}
D~Lee and Arthur K.~Lin.
\newblock Computational complexity of art gallery problems.
\newblock {\em IEEE Transactions on Information Theory}, 32:276--282, 03 1986.
\newblock \href {http://dx.doi.org/10.1109/TIT.1986.1057165}
  {\path{doi:10.1109/TIT.1986.1057165}}.

\bibitem{Lee83}
D.~T. Lee.
\newblock Visibility of a simple polygon.
\newblock {\em Computer Vision, Graphics, and Image Processing},
  22(2):207--221, 1983.
\newblock URL: \url{https://doi.org/10.1016/0734-189X(83)90065-8}, \href
  {http://dx.doi.org/10.1016/0734-189X(83)90065-8}
  {\path{doi:10.1016/0734-189X(83)90065-8}}.

\bibitem{MR09}
Nabil~Hassan Mustafa and Saurabh Ray.
\newblock {PTAS} for geometric hitting set problems via local search.
\newblock In {\em Proceedings of the 25th Annual Symposium on Computational
  Geometry}, pages 17--22. ACM, 2009.

\bibitem{O'rourke87}
Joseph O'rourke.
\newblock {\em Art gallery theorems and algorithms}, volume~57.
\newblock Oxford University Press Oxford, 1987.

\bibitem{ORourkeS83}
Joseph O'Rourke and Kenneth~J. Supowit.
\newblock Some {NP}-hard polygon decomposition problems.
\newblock {\em {IEEE} Trans. Information Theory}, 29(2):181--189, 1983.
\newblock URL: \url{https://doi.org/10.1109/TIT.1983.1056648}, \href
  {http://dx.doi.org/10.1109/TIT.1983.1056648}
  {\path{doi:10.1109/TIT.1983.1056648}}.

\bibitem{SS88}
J{\"{o}}rg{-}R{\"{u}}diger Sack and Subhash Suri.
\newblock An optimal algorithm for detecting weak visibility of a polygon
  (preliminary version).
\newblock In {\em {STACS} 88, 5th Annual Symposium on Theoretical Aspects of
  Computer Science, Bordeaux, France, February 11-13, 1988, Proceedings}, pages
  312--321, 1988.
\newblock URL: \url{https://doi.org/10.1007/BFb0035855}, \href
  {http://dx.doi.org/10.1007/BFb0035855} {\path{doi:10.1007/BFb0035855}}.

\end{thebibliography}

\end{document}